\documentclass[conference]{IEEEtran}

\usepackage{cite}
\usepackage{textcomp}
\usepackage{amsmath, amssymb, amsfonts}
\usepackage{accents}
\usepackage{color}
\usepackage{graphicx}
\usepackage{enumerate}
\usepackage{booktabs}
\usepackage{subcaption}
\usepackage{mathrsfs}
\usepackage{mathtools}
\usepackage{dsfont}
\usepackage{relsize}
\usepackage{cmtt}

\usepackage{algorithm}
\usepackage{algpseudocode}

\makeatletter
\let\NAT@parse\undefined
\makeatother
\usepackage{hyperref}  
\hypersetup{
    colorlinks= true,
    citecolor= black,
    linkcolor= black, 
    pdfborder= {0 0 0},
}

\DeclareMathOperator{\tr}{tr}
\DeclareMathOperator{\minimize}{minimize}

\DeclareMathOperator{\E}{\mathsf{E}}
\DeclareMathOperator{\Cov}{\mathsf{cov}}

\DeclareMathOperator{\sense}{\mathsf{sense}}
\DeclareMathOperator{\comm}{\mathsf{comm}}

\newtheorem{definition}{Definition}

\newtheorem{theorem}{Theorem}
\newtheorem{proposition}{Proposition}

\begin{document}

\title{Feedback Control via Integrated Sensing and Communication: Uncertainty Optimisation}
\author{
Touraj Soleymani, Mohamad Assaad, and John S. Baras
\thanks{A preliminary version of this work has been submitted to the IEEE SPAWC 2026 conference. T.~Soleymani is with the University of London, United Kingdom (e-mail: {\tt\footnotesize touraj.soleymani@citystgeorges.ac.uk}). M.~Assaad is with the CentraleSup\'{e}lec, University of Paris-Saclay, France (e-mail: {\tt\footnotesize mohamad.assaad@centralesupelec.fr}). J.~S.~Baras is with the University of Maryland College Park, United States (e-mail: {\tt\footnotesize baras@umd.edu}).}%
}
\maketitle

\begin{abstract}
This paper studies integrated sensing and communication (ISAC) coordination for feedback control tasks under shared platform constraints. We consider a cyber-physical system in which a remote dynamical process (i.e., remote source) is regulated with the support of an ISAC-enabled base station that alternates between sensing the source state and communicating control-relevant information to the source, with the two operations semantically intertwined rather than serving separate targets and users. For a Gauss--Markov source with Bernoulli-distributed sensing and communication links and a finite-horizon linear-quadratic-Gaussian (LQG) cost, we derive the optimal ISAC and control policies. Under a Bellman-operator condition, we prove that the optimal ISAC policy at the base station follows an order-threshold structure in terms of the source and base-station estimation covariances, while the optimal control policy at the source follows a certainty-equivalent structure in terms of the source state estimate. We show that the threshold region, defined as the set of estimation covariance pairs for which communication is preferred over sensing, expands with increasing source uncertainty and contracts with increasing base-station uncertainty. Our numerical analysis validates the theoretical findings.
\end{abstract}


\section{Introduction}
Integrated sensing and communication (ISAC) has recently emerged as a transformative paradigm for future wireless systems. This integration is particularly promising for cyber-physical systems, where dynamical processes are monitored through sensors, coordinated through communication networks, and influenced by actuators. In such systems, sensing, communication, and control are inherently coupled, and the overall system performance depends on their interaction within a feedback loop. Conventional methodologies typically treat these aspects separately: sensing is modelled as an independent measurement unit, communication is abstracted as a channel with prescribed constraints, and control is designed assuming a fixed information pattern. While this modular approach simplifies analysis and design, it is inappropriate for ISAC-enabled feedback control tasks, where sensing and communication compete for limited resources.

Note that ISAC can be studied at several levels, including waveform-level integration through joint signalling design, resource-level integration through shared resource design, and decision-level integration through unified operational policy design. This paper mainly focuses on integration at the decision level, while abstracting the underlying physical interaction through sensing and communication reliability and resource parameters. We consider an ISAC-enabled base station that supports a remote source by dynamically choosing between sensing the source state and communicating control-relevant information to the source, with the two operations semantically intertwined rather than serving separate targets and users. Both sensing and communication are unreliable and costly. The objective is to optimise a finite-horizon control cost at the source, together with sensing and communication overheads. This setting captures the fundamental sensing-communication tradeoff and serves as a rigorous analytical foundation for understanding the dynamics of timely information acquisition in ISAC-enabled cyber-physical systems.

\section{Related Work}
This work lies at the intersection of ISAC and sensing and communication strategies for feedback control tasks. We review the related literature below.

\subsubsection{Integrated Sensing and Communication} ISAC has recently been identified by the International Telecommunication Union as one of the six usage scenarios for 6G~\cite{wp5d2023draft}. Its emergence is driven by spectrum scarcity, shared sensing-communication hardware, and the growing role of perceptive networks in which sensing is treated as a primary service~\cite{Liu2022JSAC, liu2020joint, luong2020resource}. Research on ISAC can be organised according to the level at which sensing and communication are integrated. At the waveform level, where integration is achieved by designing the dual-functional signal~\cite{mishra2019toward, hua2023mimo, xiong2023fundamental, sturm2011waveform, liu2021cramer}, the design variables include waveforms, transmit precoding and beamforming matrices, and spatial beampatterns, while the objective is to balance sensing signal-quality metrics against those for communication. At the resource level, where integration is achieved by partitioning a shared, finite resource budget between the two functions~\cite{xu2022robust, he2023full, shi2024beamforming, liao2024power, peng2024traj, liu2024uav_iot_isac}, the design variables include transmit power, subcarrier assignments, and antenna or subarray partitions, while the objective is to ensure feasibility and efficiency under shared-resource contention. A more task-oriented perspective is decision-level integration, which treats sensing and communication as actions to support downstream objectives such as monitoring, navigation, and control. However, much of the ISAC literature remains centred on traditional performance metrics such as detection probability, achievable rate, the Cram\'er--Rao bound (CRB), and signal-to-interference-plus-noise ratio (SINR). These metrics are fundamental for physical-layer design, but they do not directly capture the requirements of feedback control tasks. In contrast to the above physical-layer and spectrum-management contributions, our work focuses on the control-theoretic implications of ISAC.

\subsubsection{Sensing and Communication Strategies}
Sensor scheduling and sensor selection for control and estimation have a long history, beginning with optimal measurement strategies for linear stochastic systems~\cite{kushner1964, meier1967, athans1972}. More recent contributions have studied large-scale and networked settings, often exploiting the notion of submodularity to derive tractable policies~\cite{gupta2006stochastic, farokhi2014stochastic, hashemi2020randomized, chamon2020approximate}. These works primarily focus on deciding when and which sensors to activate in order to improve estimation accuracy. In contrast, in our setting, uncertainty is reduced through the decisions of an ISAC system. This induces a different trade-off that is not captured by the above models. Furthermore, communication scheduling and event-triggered sampling for estimation and control have been studied extensively. Early results showed that event-triggered sampling can outperform periodic sampling under communication constraints~\cite{Astrom:2002eg}. Subsequent contributions have developed transmission policies that often exhibit threshold-type structures~\cite{gupta2009d, imer2010, wu2013, nayyar2013decentralized, lipsa2011, leong2017, voi, voi2, erasure2023, touraj2024tit, uysal2022semantic}. Our work departs fundamentally from these formulations. Rather than abstracting sensing as exogenous and optimising communication decisions in isolation, we consider a unified decision-making problem in which sensing and communication are mutually exclusive actions of a single ISAC system resource.

\subsection{Contributions and Outline}
In this paper, we study how sensing and communication in an ISAC-enabled cyber-physical system should be coordinated to support feedback control of a remote source under shared platform constraints. This setting raises two key questions: (i) how sensing and communication resources should be scheduled over time, and (ii) how control actions should be applied under sporadic information availability. For a Gauss--Markov source with Bernoulli-distributed sensing and communication links and a finite-horizon linear-quadratic-Gaussian (LQG) cost, we derive the optimal ISAC and control policies. We prove that the optimal ISAC policy at the base station follows an order-threshold structure in terms of the source and base-station estimation covariances, while the optimal control policy at the source follows a certainty-equivalent structure in terms of the source state estimate. We show that the threshold region, defined as the set of estimation covariance pairs for which communication is preferred over sensing, expands with increasing source uncertainty and contracts with increasing base-station uncertainty. Our numerical analysis validates the theoretical findings. The remainder of the paper is organised as follows. Sections~\ref{sec:formulation} and \ref{sec:main_results} formulate the problem and present the main results. Section~\ref{sec:simulation} provides numerical results. Finally, Section~\ref{sec:conclusion} concludes the paper and outlines directions for future research.

\begin{figure}[t]
\centering
  \includegraphics[width=1\linewidth]{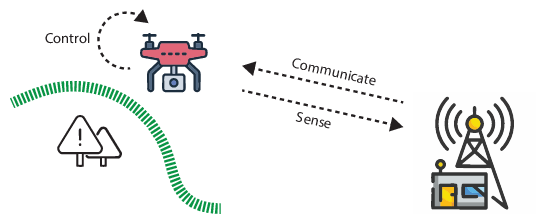}
  \caption{Illustration of an unmanned vehicle leveraging ISAC support from a base station to enhance situational awareness when onboard sensing is limited.}
  \label{fig:sketch}
\end{figure}

\section{Problem Statement}\label{sec:formulation}
We consider a cyber-physical system composed of a remote source and an ISAC-enabled base station (see Fig.~\ref{fig:sketch} for a representative example). The remote source evolves over time, yet its limited onboard sensing may prevent it from directly observing its navigational state or obtaining a full overview of the surrounding environment. The base station performs remote sensing of the source state and communicates control-relevant information back to the source. The ISAC system at each time selects either sensing or communication mode. The goal is to minimise a control cost at the remote source together with the sensing and communication operational costs over a finite time horizon of length $N$.

\subsection{System Model}
Let the remote source (i.e., remote dynamical process) be described by the Gauss--Markov model
\begin{align}
	x_{k+1} &= A x_{k} + B a_k + w_{k}, \label{eq:sys}
\end{align}
for $k \in \mathbb{N}_{[0,N]}$ with initial condition $x_{0}$, where $x_{k} \in \mathbb{R}^n$ is the state of the process, $A \in \mathbb{R}^{n \times n}$ is the state matrix, $B \in \mathbb{R}^{n \times m}$ is the input matrix, $a_k \in \mathbb{R}^m$ is the control action implemented by an actuator at the source, and $w_{k} \in \mathbb{R}^n$ is a Gaussian white noise with zero mean and covariance $W \succ 0$.

The mode of the ISAC system at the base station at each time~$k$ is determined by the mode action~$u_k$ such that $u_k =0$ if the ISAC system is in the sensing mode and $u_k =1$ if the it is in the communication mode. Both sensing and communication links are lossy and modelled as i.i.d.\ Bernoulli processes with success probabilities $\lambda_s\in [0,1]$ and $\lambda_c\in [0,1]$. Let  $\gamma_k \in \{0,1\}$ be a link success indicator at time~$k$. We can write
\begin{align}
\Pr(\gamma_k = 1 \mid u_k) =
\begin{cases}
\lambda_s, & \text{if } u_k= 0,\\[1\jot]
\lambda_c, & \text{if } u_k=1,
\end{cases}
\end{align}
and $\Pr(\gamma_k = 0 \mid u_k)=1-\Pr(\gamma_k = 1 \mid u_k)$. When $u_k =0$, the ISAC system attempts to sense a noisy version of the output of the source. In this case, the measurement at time $k$ is given~by
\begin{align}
y_{k+1} = \left\{
  \begin{array}{l l}
     C x_{k+1} + v_{k+1}, & \ \text{if} \ u_{k} = 0 \wedge \gamma_{k} = 1, \\[1\jot]
     \varnothing , & \ \text{otherwise},
  \end{array} \right.	
\end{align}
for $k \in \mathbb{N}_{[0,N]}$ with $y_{0} = C x_0 + v_0$ by convention, where $y_{k} \in \mathbb{R}^p$ is the output of the radar at the base station, $C \in \mathbb{R}^{p \times n}$ is the output matrix, $v_{k} \in \mathbb{R}^p$ is a Gaussian white noise with zero mean and covariance $V \succ 0$, and $\varnothing$ represents absence of data. However, when $u_k =1$, the ISAC system attempts to transmit control-relevant information (i.e., the best local state estimate) to the source. Let $\check{x}_{k}$ be the MMSE state estimate at the base station and $Q_k$ be its associated covariance matrix. In this case, the received message at time $k$ is given~by
\begin{align}\label{eq:channel}
z_{k+1} = \left\{
  \begin{array}{l l}
     (\check{x}_{k},Q_{k}) , & \ \text{if} \ u_{k} = 1 \wedge \gamma_{k} = 1, \\[1\jot]
     \varnothing , & \ \text{otherwise},
  \end{array} \right.	
\end{align}
for $k \in \mathbb{N}_{[0,N]}$ with $z_{0} = \varnothing$ by convention, where $z_{k}$ is the output of the channel and $\varnothing$ represents absence of data. 


We assume that the following assumptions are satisfied: the initial condition $x_0$ is a Gaussian vector with mean $m_0$ and covariance $M_0$; the random variables $x_0$, $w_t$, and $v_s$ for $t,s \in \mathbb{N}_{[0,N]}$ are mutually independent; the communication success rate is higher than the sensing success rate, i.e., $\lambda^c \geq \lambda^s$; a packet acknowledgement containing the success indicator $\gamma_k$ associated with the action $u_k$ is available to the base station before the next decision epoch; and the error due to quantisation is negligible.

\subsection{Decision-Making Problem}
In our formulation, the ISAC at the base station and the actuator at the source act as two distributed decision makers. Define the information sets $\mathcal{I}_k^{b} := \{ u_{0:k-1}, \gamma_{0:k-1} \}$ and $\mathcal{I}_k^{s} := \{ u_{0:k-1}, \gamma_{0:k-1}, z_{0:k-1} \}$. An admissible ISAC policy is a sequence of measurable mappings $\pi = \{\pi_k\}_{k=0}^N$, where $\pi_k : \mathcal{I}_k^{b} \to \{0,1\}$. The mode action at time $k$ is given by $u_k = \pi_k(\mathcal{I}_k^{b})$. Moreover, an admissible control policy is a sequence of measurable~mappings $\mu = \{\mu_k\}_{k=0}^N$, where $\mu_k : \mathcal{I}_k^{s} \to \mathbb{R}^m$. The control action at time $k$ is given by $a_k = \mu_k(\mathcal{I}_k^{s})$. Note that the MMSE state estimate at time $k$ at the base station is obtained based on $\mathcal{I}_k^{b}$ and all raw measurements $y_{0:k}$. This state estimate represents the best reconstruction of the source state, which can be communicated to the source. However, an admissible ISAC policy depends only on $\mathcal{I}_k^{b}$, which leads to a tractable policy that relies on estimation uncertainties.

Our goal is to identify the best solution $(\pi^\star,\mu^\star)$ to the stochastic optimisation problem
\begin{align}\label{eq:main_problem}
	\underset{\pi \in \mathcal{P}, \mu \in \mathcal{M}}{\minimize} \ \Upsilon(\pi,\mu),
\end{align}
where $\mathcal{P}$ and $\mathcal{M}$ are the sets of admissible ISAC policies and admissible control policies, and 
\begin{align}\label{eq:loss-function}
\Upsilon(\pi,\mu) &:= \frac{1}{N+1} \E \bigg [\sum_{k=0}^{N+1} x_{k}^T \Omega^x_{k} x_{k} + \sum_{k=0}^{N} a_k^T \Omega^a_{k} a_k \bigg] \notag\\
&\qquad + \frac{1}{N+1} \E \bigg [\sum_{k=0}^{N} \theta^s (1-u_k) + \theta^c u_k \bigg]
\end{align}
for $\Omega^x_{k} \succeq 0$ and $\Omega^a_{k} \succ 0$ as weighting matrices and $\theta^c_k>0$ and $\theta^s>0$ as weighting scalars. We assume that $\theta^c \geq \theta^s$, meaning that the cost of communication is at least equal to that of sensing.

\section{Main Results}\label{sec:main_results}
In this section, we present our main results. We first express the optimal estimators at the base station and the source.

\begin{proposition}\label{prop:1}
Let $\check{x}_k = \E[x_k | \mathcal{I}^b_k, y_{0:k}]$ and $Q_k = \Cov[x_k - \check{x}_k | \mathcal{I}^b_k, y_{0:k}]$. The optimal estimator at the base station satisfies the recursive equation
\begin{align}
	\check{x}_{k+1} &= A \check{x}_k + B a_k + \gamma_{k} (1-u_k) K_{k+1} \nu_{k+1},\\[1\jot]
	Q_{k+1} &= \big( (A Q_k A^T \!+\! W)^{-1} \!+ \gamma_{k} (1 \!-\! u_k) C^T V ^{-1} C \big)^{-1} ,
\end{align}
for $k \in \mathbb{N}_{[0,N]}$ with initial conditions $\check{x}_{0} = m_{0} + K_{0} (y_0 - Cm_0)$, $Q_{0} = (M_{0}^{-1} + C^T V^{-1} C)^{-1}$, where $\nu_k = y_k - C A \check{x}_{k-1} - C B a_{k-1}$ and $K_{k} = Q_{k} C^T V^{-1}$. Let $\hat{x}_k = \E[x_k | \mathcal{I}^s_k]$ and $P_k = \Cov[x_k - \hat{x}_k | \mathcal{I}^s_k]$. The optimal estimator at the source satisfies the recursive equation
\begin{align}
	\hat{x}_{k+1} &= A \hat{x}_k + B a_k + \gamma_{k} u_{k} A \tilde{e}_k,\\[1\jot]
	P_{k+1} &= A P_k A^T + \gamma_{k} u_{k} A (Q_k - P_k) A^T + W,
\end{align}
for $k \in \mathbb{N}_{[0,N]}$ with initial conditions $\hat{x}_0 = m_0$ and $P_0 = M_0$, where $\tilde{e}_k = \check{x}_k - \hat{x}_k$.
\end{proposition}

Note that, according to Proposition~\ref{prop:1}, the base-station estimator follows the usual prediction-correction form of a Kalman filter with intermittent observations. Its covariance recursion is written in information form: the prior covariance \(AQ_kA^\top+W\) is corrected by the measurement information \(C^\top V^{-1}C\) only when \(\gamma_k(1-u_k)=1\), i.e., when a sensing operation is successful. By contrast, the source estimator cannot use raw measurements directly. When \(\gamma_ku_k=1\), i.e., when the source receives the base-station estimate, its covariance is reduced from the predicted source covariance \(AP_kA^\top+W\) to the propagated base-station covariance \(AQ_kA^\top+W\).

Since the covariance recursions are deterministic functions of the mode and acknowledgement history, the ISAC policy can equivalently be represented as a function of the covariance state \((P_k,Q_k)\). Let $\mathcal D := \{(P,Q):P,Q\in\mathbb S_+^n,\ 0\preceq Q\preceq P\}$. The prediction and update maps are $\Phi(X) =AXA^\top+W$ and $\Psi(X) = \left(\Phi(X)^{-1}+C^\top V^{-1}C\right)^{-1}$. Define $\Gamma_k := A^\top S_{k+1} B \big(B^\top S_{k+1} B + \Omega^a_k\big)^{-1} B^\top S_{k+1} A$, where $S_t \in\mathbb S_+^n$ satisfies the Riccati recursion
\begin{align}\label{eq:riccati}
S_t &= \Omega^x_t + A^T S_{t+1} A - A^T S_{t+1} B \nonumber\\[1\jot]
	&\qquad \qquad \times \big(B^T S_{t+1} B + \Omega^a_t \big)^{-1} B^T S_{t+1} A.
\end{align}
with terminal condition \(S_{N+1}=\Omega^x_{N+1}\). Let us define a reduced value function as
\begin{align*}
V_k(P,Q) &:= \min_{\pi} \mathbb \E \Big[ \sum_{t=k}^{N} ( \tr(\Gamma_tP_t) + \bar{\theta} u_t)  \mid (P_k,Q_k)=(P,Q) \Big],
\end{align*}
where $\bar{\theta} = \theta^c - \theta^s$  with terminal value $V_{N+1}(P,Q)=0$. Equivalently, $V_k=T_kV_{k+1}$, for $k\in \mathbb{N}_{[0,N]}$, where, for a generic continuation function \(V:\mathcal D\to\mathbb R\),
\[
(T_kV)(P,Q) = \operatorname{tr}(\Gamma_kP) + \min\{F_{k,\sense}^{V}(P,Q),F_{k,\comm}^{V}(P,Q)\}.
\]
Here,
\begin{align*}
F_{k,\sense}^{V}(P,Q) &= \lambda^sV(\Phi(P),\Psi(Q))\\
&\quad + (1-\lambda^s)V(\Phi(P),\Phi(Q)),\\
F_{k,\comm}^{V}(P,Q) &= \bar{\theta} + \lambda^cV(\Phi(Q),\Phi(Q))\\
&\quad + (1-\lambda^c)V(\Phi(P),\Phi(Q)).
\end{align*}
For the optimal value function, we write $F_{k,u}:=F_{k,u}^{V_{k+1}}$, for $u\in\{\sense,\comm\}$. The switching advantage function is
\[
\Delta_k(P,Q) := F_{k,\sense}^{V_{k+1}}(P,Q) - F_{k,\comm}^{V_{k+1}}(P,Q).
\]

For the structural result, we impose a local operator condition on the
continuation value. For every admissible triple $0\preceq Q\prec Q'\preceq P$, and every \(u\in\{ \sense,\comm \}\), define
\begin{align*}
A_{k,u}^{V} := F_{k,u}^{V}(\Phi(Q'),\Phi(Q')) - F_{k,u}^{V}(\Phi(Q),\Phi(Q)),\\
B_{k,u}^{V} := F_{k,u}^{V}(\Phi(P),\Psi(Q')) - F_{k,u}^{V}(\Phi(P),\Psi(Q)),\\
C_{k,u}^{V} := F_{k,u}^{V}(\Phi(P),\Phi(Q')) - F_{k,u}^{V}(\Phi(P),\Phi(Q)).
\end{align*}
Let \(d_\Phi(P,Q,Q')>0\) and \(d_\Psi(P,Q,Q')>0\) be positive gauges, and define
\begin{align*}
\kappa_{\Psi,\Phi}(P,Q,Q') &:= \frac{d_\Psi(P,Q,Q')}{d_\Phi(P,Q,Q')},\\
\eta_k^\Phi(P,Q,Q') &:= \frac{\operatorname{tr}(\Gamma_k (\Phi(Q')-\Phi(Q))) }{d_\Phi(P,Q,Q')}.
\end{align*}

\begin{definition}[Local Lipschitz Condition]
We say that $V$ satisfies the local Lipschitz condition if, for every $0\preceq Q\prec Q'\preceq P$, there exist interval-local one-sided slope bounds $L_{k,u}^{A,V}$, $L_{k,u}^{B,V}$, and $L_{k,u}^{C,V}$, for $u\in\{\sense,\comm\}$, such that $A_{k,u}^{V}\ge L_{k,u}^{A,V}d_\Phi$, $B_{k,u}^{V}\le L_{k,u}^{B,V}d_\Psi$, $C_{k,u}^{V}\le L_{k,u}^{C,V}d_\Phi$, and
\begin{align*}
&\lambda^c \eta_k^\Phi+ \lambda^c\min_{u\in\{\sense,\comm\}}L_{k,u}^{A,V} \\
&\qquad \ge \lambda^s \kappa_{\Psi,\Phi} \max_{u\in\{\sense,\comm\}}L_{k,u}^{B,V} + (\lambda^c-\lambda^s) \max_{u\in\{\sense,\comm\}}L_{k,u}^{C,V},
\end{align*}
where all slope bounds and gauges are evaluated at \((P,Q,Q')\).
\end{definition}

\begin{theorem}\label{thm:1}
Suppose that $V_{k+1}$ satisfies the local Lipschitz condition for $k\in \mathbb{N}_{[0,N]}$. Associated with the problem in \eqref{eq:main_problem}, the optimal ISAC policy obeys an order-threshold structure in the following sense: the switching region
\begin{align}
\mathcal C_k:=\{(P,Q)\in\mathcal D:\Delta_k(P,Q)\ge0\}
\end{align}
is upward closed in \(P\) and downward closed in \(Q\) under the Loewner order, i.e.,
\begin{align}
(P,Q)\in\mathcal C_k \ &\Longrightarrow \ (P',Q)\in\mathcal C_k,\\
(P,Q')\in\mathcal C_k \ &\Longrightarrow \ (P,Q)\in\mathcal C_k.
\end{align}
for $Q\preceq Q'\preceq P \preceq P'$; and the optimal control policy obeys a certainty-equivalent structure, i.e.,
\begin{align}
a_k^\star=-L_k\hat x_k,
\end{align}
for $k \in \mathbb{N}_{[0,N]}$, where $L_k = (B^T S_{k+1} B + \Omega^a_k)^{-1} B^T S_{k+1} A$.
\end{theorem}

Note that Theorem~\ref{thm:1} separates the roles of the actuator and the ISAC system. The control action follows the standard LQG form and depends on the source-side state estimate \(\hat x_k\). The ISAC decision, by contrast, depends on the covariance state \((P,Q)\). The order-threshold property implies that communication becomes more attractive as the source covariance \(P\) increases, while sensing becomes more attractive as the base-station covariance \(Q\) increases. In the scalar case, where the order is total, the result reduces to the usual threshold policy. However, in the matrix-valued case, the result should be interpreted as a monotone switching region under the Loewner order, rather than as a scalar threshold curve. 

We now turn to the proof of Theorem~\ref{thm:1}. We first define two structural classes of functions.
\begin{definition}[Class \(\mathcal M\)]
A function \(f:\mathcal D\to\mathbb R\) belongs to \(\mathcal M\) if it is nondecreasing in each covariance argument under the Loewner order, i.e.,
\begin{align*}
P'\succeq P \ \Longrightarrow \ f(P',Q)\ge f(P,Q),\\
Q'\succeq Q \ \Longrightarrow \ f(P,Q')\ge f(P,Q).
\end{align*}
\end{definition}

\begin{definition}[Class \(\mathcal W\)]
A function \(f:\mathcal D\to\mathbb R\) belongs to \(\mathcal W\) if, for every $0\preceq Q\preceq Q'\preceq P$, we have
\[
\lambda^c A_f(P,Q,Q') \ge \lambda^s B_f(P,Q,Q') + (\lambda^c-\lambda^s)C_f(P,Q,Q'),
\]
where
\begin{align*}
A_f(P,Q,Q') := f(\Phi(Q'),\Phi(Q'))-f(\Phi(Q),\Phi(Q)),\\
B_f(P,Q,Q') := f(\Phi(P),\Psi(Q'))-f(\Phi(P),\Psi(Q)),\\
C_f(P,Q,Q') := f(\Phi(P),\Phi(Q'))-f(\Phi(P),\Phi(Q)).
\end{align*}
\end{definition}
We define the structural class $\mathcal H:=\mathcal M\cap\mathcal W$. We now provide the proof of Theorem~\ref{thm:1}.
\begin{IEEEproof}
The proof is structured in four steps.

\noindent \textit{Step 1:}
We first exploit the separation structure of the problem. For a fixed ISAC policy $\pi$, define the value function
\begin{align}
\Upsilon'(\pi,\mu) := \E \bigg[ \sum_{k=0}^{N} h_k^T \Lambda_k h_k + \bar{\theta} u_k \bigg],
\label{eq:Vd-def-proof}
\end{align}
where $h_t := a_t + (B^\top S_{t+1} B + \Omega^a_t)^{-1} B^\top S_{t+1} A x_t$ and $\Lambda_t := B^\top S_{t+1} B + \Omega^a_t \succ 0$. Applying the separation principle, one can show that the control policy over $\mu \in \mathcal{M}$ that optimises \eqref{eq:Vd-def-proof} is $a_k^\star = -L_k \hat{x}_k$, which is a certainty-equivalent policy. Substituting this result into the cost function yields the equivalent cost function
\begin{align}
\Upsilon'(\pi,\mu^\star) &= \E \bigg[ \sum_{k=0}^{N} \big( \tr(\Gamma_k P_k) + \bar{\theta} u_k \big) \bigg].
\label{eq:reduced-cost}
\end{align}
Thus, the joint problem reduces to an estimation problem, expressed as a minimisation of a weighted sum of estimation covariances. 

\noindent \textit{Step 2:} We show that under the local Lipschitz operator condition, if $V_{k+1}\in\mathcal H$, then $V_k\in\mathcal H$. Since \(V_{N+1}=0\), we have \(V_{N+1}\in\mathcal H\). We prove separately that \(V_k\in\mathcal M\) and \(V_k\in\mathcal W\). Since \(V_{k+1}\in\mathcal M\), and since \(\Phi\) and \(\Psi\) are
Loewner monotone, the functions
\begin{align*}
(P,Q) &\mapsto V_{k+1}(\Phi(P),\Psi(Q)),\\
(P,Q) &\mapsto V_{k+1}(\Phi(P),\Phi(Q)),\\
(P,Q) &\mapsto V_{k+1}(\Phi(Q),\Phi(Q))
\end{align*}
are nondecreasing on \(\mathcal D\). Therefore \(F_{k,s}\) and \(F_{k,c}\) are nondecreasing. The pointwise minimum of nondecreasing functions is nondecreasing. Since $P\mapsto \operatorname{tr}(\Gamma_k P)$ is nondecreasing under the Loewner order for \(\Gamma_k\succeq0\), it follows that $V_k\in\mathcal M$.

It remains to prove \(V_k\in\mathcal W\). Fix $0\preceq Q\preceq Q'\preceq P$. Define
\begin{align*}
A_{V_k} &:= V_k(\Phi(Q'),\Phi(Q')) - V_k(\Phi(Q),\Phi(Q)),\\
B_{V_k} &:= V_k(\Phi(P),\Psi(Q')) - V_k(\Phi(P),\Psi(Q)),\\
C_{V_k} &:= V_k(\Phi(P),\Phi(Q')) - V_k(\Phi(P),\Phi(Q)).
\end{align*}
Using $V_k(P,Q) = \operatorname{tr}(\Gamma_kP)+\min_{u\in\{\sense,\comm\}}F_{k,u}(P,Q)$, we have
\begin{align*}
A_{V_k} &=\operatorname{tr}\!\left(\Gamma_k[\Phi(Q')-\Phi(Q)]\right)
\\
&\quad+ \min_uF_{k,u}(\Phi(Q'),\Phi(Q')) - \min_uF_{k,u}(\Phi(Q),\Phi(Q)).
\end{align*}
Since $\min_u x_u-\min_u y_u \ge \min_u(x_u-y_u)$, we get
\begin{align*}
A_{V_k} \ge \operatorname{tr}\!\left(\Gamma_k[\Phi(Q')-\Phi(Q)]\right) + \min_{u\in\{\sense,\comm\}} A_{F_{k,u}}.
\end{align*}
Similarly, since $\min_u x_u-\min_u y_u \le \max_u(x_u-y_u)$, and since the immediate term cancels in the \(B\)- and \(C\)-increments, we get $B_{V_k} \le \max_{u\in\{\sense,\comm\}} B_{F_{k,u}}$, and $C_{V_k} \le \max_{u\in\{\sense,\comm\}} C_{F_{k,u}}$. By the interval-local bounds, $\min_u A_{F_{k,u}} \ge d_\Phi\min_u \underline L_{k,u}^A$, $\max_u B_{F_{k,u}} \le d_\Psi\max_u \overline L_{k,u}^B = d_\Phi\kappa_{\Psi,\Phi}\max_u \overline L_{k,u}^B$, and $\max_u C_{F_{k,u}} \le d_\Phi\max_u \overline L_{k,u}^C$. Also, observe that $\operatorname{tr}\!\left(\Gamma_k[\Phi(Q')-\Phi(Q)]\right) = d_\Phi\eta_k^\Phi$. Combining these inequalities with the local Lipschitz condition gives
\[
\lambda^cA_{V_k}\ge\lambda^sB_{V_k}+(\lambda^c-\lambda^s)C_{V_k}.
\]
Thus, $V_k\in\mathcal W$. Consequently, $V_k\in\mathcal M\cap\mathcal W=\mathcal H$.

\noindent \textit{Step 3:} We now prove that $\Delta_k(P,Q)$ is nondecreasing in $P$ and nonincreasing in $Q$. Let \(f:=V_{k+1}\). First, take \(P'\succeq P\). Since \(\Phi\) is Loewner monotone, $\Phi(P')\succeq\Phi(P)$. Using the expression
\begin{align*}
\Delta_k(P,Q) &= -\bar{\theta} + \lambda^s f(\Phi(P),\Psi(Q)) \\
&\quad+ (\lambda^c-\lambda^s)f(\Phi(P),\Phi(Q)) - \lambda^c f(\Phi(Q),\Phi(Q)),
\end{align*}
we get
\begin{align*}
&\Delta_k(P',Q)-\Delta_k(P,Q)\\
&= \lambda^s \Big[ f(\Phi(P'),\Psi(Q))-f(\Phi(P),\Psi(Q)) \Big]\\
&+ (\lambda^c-\lambda^s) \Big[ f(\Phi(P'),\Phi(Q))-f(\Phi(P),\Phi(Q)) \Big].
\end{align*}
Because \(f\in\mathcal M\), both finite differences are nonnegative. Since $\lambda^s>0$ and $\lambda^c-\lambda^s\ge0$, we obtain $\Delta_k(P',Q)\ge\Delta_k(P,Q)$.

Second, take \(0\preceq Q\preceq Q'\preceq P\). Define $A:=A_f(P,Q,Q')$, $B:=B_f(P,Q,Q')$, $C:=C_f(P,Q,Q')$. Then
\begin{align*}
\Delta_k(P,Q')-\Delta_k(P,Q) = \lambda^s B+(\lambda^c-\lambda^s)C-\lambda^c A.
\end{align*}
Because \(f\in\mathcal W\), $\lambda^c A \ge \lambda^s B+(\lambda^c-\lambda^s)C$. Therefore, $\Delta_k(P,Q')-\Delta_k(P,Q)\le0$. This proves the claim.

\noindent \textit{Step 4:} Therefore, by Step 3, $\Delta_k$ is nondecreasing in $P$ and nonincreasing in $Q$. Now suppose $(P,Q)\in\mathcal C_k$ and $P'\succeq P$. Then, $\Delta_k(P',Q)\ge\Delta_k(P,Q)\ge0$, so $(P',Q)\in\mathcal C_k$. Similarly, suppose \((P,Q')\in\mathcal C_k\) and $Q\preceq Q'\preceq P$. Since \(\Delta_k\) is nonincreasing in \(Q\), $\Delta_k(P,Q)\ge\Delta_k(P,Q')\ge0$, so $(P,Q)\in\mathcal C_k$. This proves the claimed order-threshold structure.
\end{IEEEproof}

\begin{figure}[t]
\centering
  \includegraphics[width=0.85\linewidth, trim= 0mm 0mm 0mm 5mm,clip]{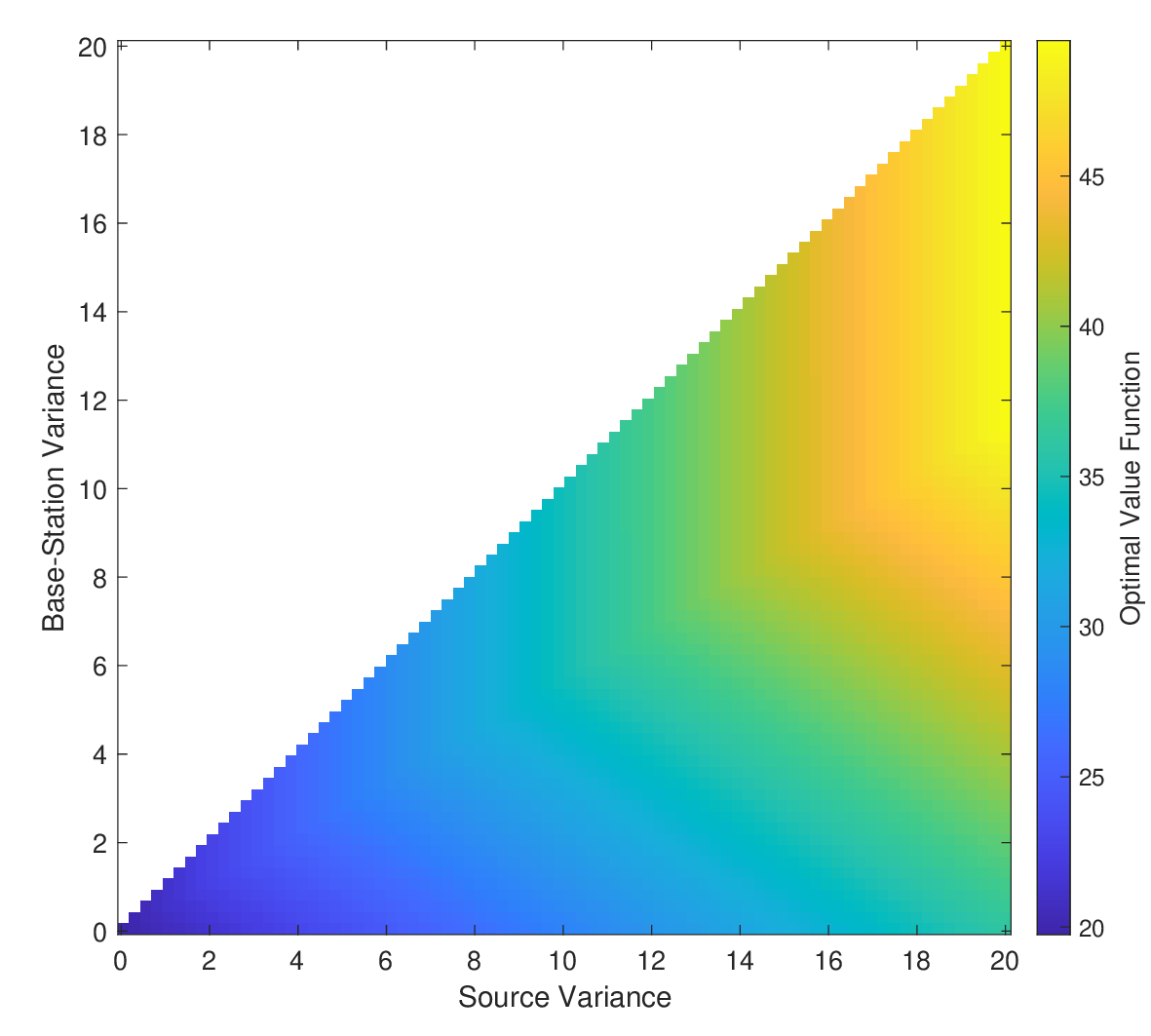}
  \caption{Value function at \(k=0\) over the source and base-station variances.}
  \label{fig:value_function}
\end{figure}

\begin{figure}[t]
\centering
  \includegraphics[width=0.85\linewidth, trim= 0mm 0mm 0mm 5mm,clip]{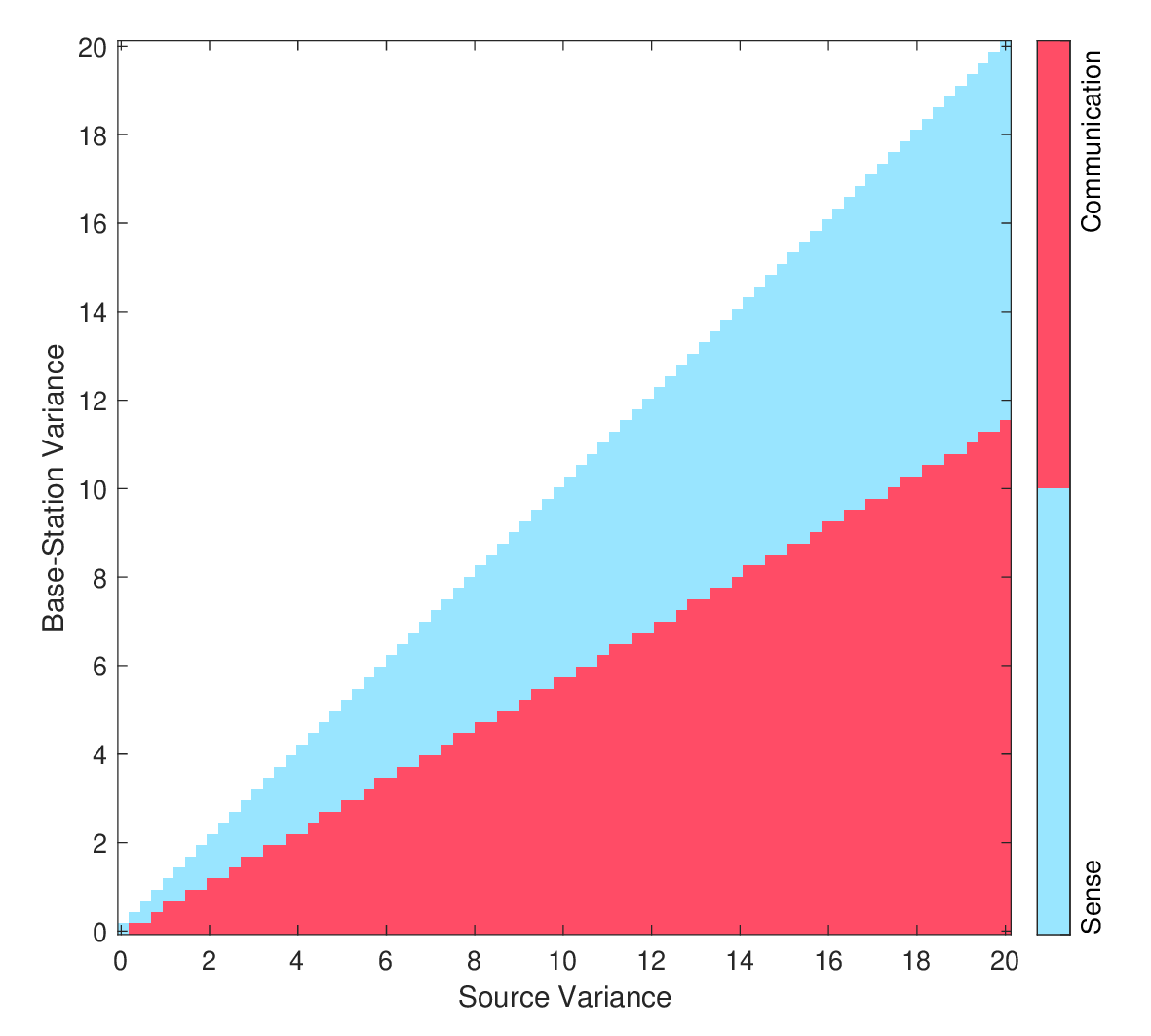}
  \caption{Optimal switching map at \(k=0\), showing the threshold structure of the optimal ISAC policy.}
  \label{fig:decision_map}
\end{figure}

\section{Numerical Experiments}\label{sec:simulation}
In this section, we present a numerical experiment. For the simulations, we consider the dynamic programming formulation with the value function corresponding to the reduced cost after applying the control policy, and use the following parameters: $A = 0.9$, $B = 1$, $C = 1.0$, $W = 0.3$, $V = 0.1$, and $\Omega^x_k = \Omega^a_k = 1$. The horizon length is chosen as $N = 50$ and link reliabilities are $(\lambda_s,\lambda_c)=(0.8,0.85)$. To compute the value function $V_0(P,Q)$ in terms of estimation variances, we evaluated the value function on a regular grid, discretised into uniformly spaced points. A sufficient resolution is chosen to reveal the qualitative structure of the value function while keeping computation tractable. Fig.~\ref{fig:value_function} illustrates the value function $V_0(P,Q)$ over the grid of source and base-station variances. As expected, the value function is monotone increasing in the source variance $P$ and the base-station variance $Q$, i.e., higher uncertainty at the source or the base station directly raises the expected quadratic cost. To complement this, Fig.~\ref{fig:decision_map} shows the corresponding decision map, indicating the regions where it is optimal to sense ($u=0$) or to communicate ($u=1$). The boundary between these two regions clearly exhibits the threshold structure predicted by our theoretical analysis, i.e., for fixed $Q$, larger values of $P$ favour communication, while for fixed $P$, larger values of $Q$ favour sensing.

\section{Conclusion}\label{sec:conclusion}
In this paper, we have studied the role of ISAC in cyber-physical systems. We characterised the structural properties of the optimal ISAC and control policies. Under a Bellman-operator condition, we established that the optimal ISAC policy at the base station admits an order-threshold structure in terms of the source and base-station estimation covariances, while the optimal control policy at the source takes the certainty-equivalent structure in terms of the source state estimate. These results provide a formal synthesis framework for ISAC-enabled remote control tasks, revealing how sensing and communication should be dynamically balanced to achieve optimal closed-loop performance.

\bibliography{../../../mybib}

\begin{thebibliography}{10}

\bibitem{wp5d2023draft}
{ITU-R}, ``Framework and overall objectives of the future development of {IMT} for 2030 and beyond,'' Tech. Rep. {Recommendation ITU-R M.2160-0}, International Telecommunication Union, 2023.

\bibitem{Liu2022JSAC}
F.~Liu, Y.~Cui, C.~Masouros, J.~Xu, T.~X. Han, Y.~Li, C.~Yuen, R.~Schober, and H.~V. Poor, ``Integrated sensing and communications: Toward dual-functional wireless networks for 6g and beyond,'' {\em IEEE Journal on Selected Areas in Communications}, vol.~40, no.~6, pp.~1728--1767, 2022.

\bibitem{liu2020joint}
F.~Liu, C.~Masouros, A.~P. Petropulu, H.~Griffiths, and L.~Hanzo, ``Joint radar and communication design: Applications, state-of-the-art, and the road ahead,'' {\em IEEE Trans. on Communications}, vol.~68, no.~6, pp.~3834--3862, 2020.

\bibitem{luong2020resource}
N.~C. Luong, X.~Lu, D.~T. Hoang, D.~Niyato, and D.~I. Kim, ``Radio resource management in joint radar and communication: A comprehensive survey,'' {\em IEEE Communications Surveys \& Tutorials}, vol.~23, no.~2, pp.~780--814, 2021.

\bibitem{mishra2019toward}
K.~V. Mishra, M.~B. Shankar, V.~Koivunen, B.~Ottersten, and S.~A. Vorobyov, ``Toward millimeter-wave joint radar communications: A signal processing perspective,'' {\em IEEE Signal Processing Magazine}, vol.~36, no.~5, pp.~100--114, 2019.

\bibitem{hua2023mimo}
H.~Hua, T.~X. Han, and J.~Xu, ``{MIMO} integrated sensing and communication: {CRB}-rate tradeoff,'' {\em IEEE Trans. on Wireless communications}, vol.~23, no.~4, pp.~2839--2854, 2023.

\bibitem{xiong2023fundamental}
Y.~Xiong, F.~Liu, Y.~Cui, W.~Yuan, T.~X. Han, and G.~Caire, ``On the fundamental tradeoff of integrated sensing and communications under {Gaussian} channels,'' {\em IEEE Trans. on Information Theory}, vol.~69, no.~9, pp.~5723--5751, 2023.

\bibitem{sturm2011waveform}
C.~Sturm and W.~Wiesbeck, ``Waveform design and signal processing aspects for fusion of wireless communications and radar sensing,'' {\em Proceedings of the IEEE}, vol.~99, no.~7, pp.~1236--1259, 2011.

\bibitem{liu2021cramer}
F.~Liu, Y.-F. Liu, A.~Li, C.~Masouros, and Y.~C. Eldar, ``Cram{\'e}r-{Rao} bound optimization for joint radar-communication beamforming,'' {\em IEEE Trans. on Signal Processing}, vol.~70, pp.~240--253, 2021.

\bibitem{xu2022robust}
D.~Xu, X.~Yu, D.~W.~K. Ng, A.~Schmeink, and R.~Schober, ``Robust and secure resource allocation for isac systems: A novel optimization framework for variable-length snapshots,'' {\em IEEE Trans. on Communications}, vol.~70, no.~12, pp.~8196--8214, 2022.

\bibitem{he2023full}
Z.~He, W.~Xu, H.~Shen, D.~W.~K. Ng, Y.~C. Eldar, and X.~You, ``Full-duplex communication for {ISAC}: Joint beamforming and power optimization,'' {\em IEEE Journal on Selected Areas in Communications}, vol.~41, no.~9, pp.~2920--2936, 2023.

\bibitem{shi2024beamforming}
Y.~Shi, L.~Li, W.~Lin, W.~Liang, and Z.~Han, ``Beamforming in secure integrated sensing and communication systems with antenna allocation,'' in {\em Proc. IEEE Globecom Workshops}, pp.~1--6, IEEE, 2024.

\bibitem{liao2024power}
B.~Liao, H.~Q. Ngo, M.~Matthaiou, and P.~J. Smith, ``Power allocation for massive {MIMO}-{ISAC} systems,'' {\em IEEE Trans. on Wireless Communications}, vol.~23, no.~10, pp.~14232--14248, 2024.

\bibitem{peng2024traj}
S.~Peng, B.~Li, L.~Liu, Z.~Fei, and D.~Niyato, ``Trajectory design and resource allocation for multi-{UAV}-assisted sensing, communication, and edge computing integration,'' {\em IEEE Trans. on Communications}, vol.~73, no.~4, pp.~2847--2861, 2024.

\bibitem{liu2024uav_iot_isac}
Z.~Liu, X.~Liu, Y.~Liu, V.~C. Leung, and T.~S. Durrani, ``{UAV} assisted integrated sensing and communications for {I}nternet of {T}hings: {3D} trajectory optimization and resource allocation,'' {\em IEEE Trans. on Wireless Communications}, vol.~23, no.~8, pp.~8654--8667, 2024.

\bibitem{kushner1964}
H.~J. Kushner, ``On the optimum timing of observations for linear control systems with unknown initial state,'' {\em IEEE Trans. on Automatic Control}, vol.~9, no.~2, pp.~144--150, 1964.

\bibitem{meier1967}
L.~Meier, J.~Peschon, and R.~M. Dressler, ``Optimal control of measurement subsystems,'' {\em IEEE Trans. on Automatic Control}, vol.~12, no.~5, pp.~528--536, 1967.

\bibitem{athans1972}
M.~Athans, ``On the determination of optimal costly measurement strategies for linear stochastic systems,'' {\em Automatica}, vol.~8, no.~4, pp.~397--412, 1972.

\bibitem{gupta2006stochastic}
V.~Gupta, T.~H. Chung, B.~Hassibi, and R.~M. Murray, ``On a stochastic sensor selection algorithm with applications in sensor scheduling and sensor coverage,'' {\em Automatica}, vol.~42, no.~2, pp.~251--260, 2006.

\bibitem{farokhi2014stochastic}
F.~Farokhi and K.~H. Johansson, ``Stochastic sensor scheduling for networked control systems,'' {\em IEEE Transactions on Automatic Control}, vol.~59, no.~5, pp.~1147--1162, 2014.

\bibitem{hashemi2020randomized}
A.~Hashemi, M.~Ghasemi, H.~Vikalo, and U.~Topcu, ``Randomized greedy sensor selection: Leveraging weak submodularity,'' {\em IEEE Transactions on Automatic Control}, vol.~66, no.~1, pp.~199--212, 2020.

\bibitem{chamon2020approximate}
L.~F. Chamon, G.~J. Pappas, and A.~Ribeiro, ``Approximate supermodularity of {Kalman} filter sensor selection,'' {\em IEEE Transactions on Automatic Control}, vol.~66, no.~1, pp.~49--63, 2020.

\bibitem{Astrom:2002eg}
K.~J. {\AA}str{\"o}m and B.~Bernhardsson, ``Comparison of {R}iemann and {L}ebesgue sampling for first order stochastic systems,'' in {\em Proc. IEEE Conf. on Decision and Control}, pp.~2011--2016, 2002.

\bibitem{gupta2009d}
V.~Gupta, A.~F. Dana, J.~P. Hespanha, R.~M. Murray, and B.~Hassibi, ``Data transmission over networks for estimation and control,'' {\em IEEE Trans. on Automatic Control}, vol.~54, no.~8, pp.~1807--1819, 2009.

\bibitem{imer2010}
O.~C. Imer and T.~Ba{\c{s}}ar, ``Optimal estimation with limited measurements,'' {\em Intl. Journal of Systems, Control and Communications}, vol.~2, no.~1-3, pp.~5--29, 2010.

\bibitem{wu2013}
J.~Wu, Q.-S. Jia, K.~H. Johansson, and L.~Shi, ``Event-based sensor data scheduling: Trade-off between communication rate and estimation quality,'' {\em IEEE Trans. on Automatic Control}, vol.~58, no.~4, pp.~1041--1046, 2013.

\bibitem{nayyar2013decentralized}
A.~Nayyar, A.~Mahajan, and D.~Teneketzis, ``Decentralized stochastic control with partial history sharing: A common information approach,'' {\em IEEE Trans. on Automatic Control}, vol.~58, no.~7, pp.~1644--1658, 2013.

\bibitem{lipsa2011}
G.~M. Lipsa and N.~C. Martins, ``Remote state estimation with communication costs for first-order {LTI} systems,'' {\em IEEE Trans. on Automatic Control}, vol.~56, no.~9, pp.~2013--2025, 2011.

\bibitem{leong2017}
A.~S. Leong, S.~Dey, and D.~E. Quevedo, ``Sensor scheduling in variance based event triggered estimation with packet drops,'' {\em IEEE Trans. on Automatic Control}, vol.~62, no.~4, pp.~1880--1895, 2017.

\bibitem{voi}
T.~Soleymani, J.~S. Baras, and S.~Hirche, ``Value of information in feedback control: Quantification,'' {\em IEEE Trans. on Automatic Control}, vol.~67, no.~7, pp.~3730--3737, 2022.

\bibitem{voi2}
T.~Soleymani, J.~S. Baras, S.~Hirche, and K.~H. Johansson, ``Value of information in feedback control: Global optimality,'' {\em IEEE Trans. on Automatic Control}, vol.~68, no.~6, pp.~3641--3647, 2023.

\bibitem{erasure2023}
T.~Soleymani, J.~S. Baras, and K.~H. Johansson, ``State estimation over delayed and lossy channels: An encoder-decoder synthesis,'' {\em IEEE Trans. on Automatic Control}, vol.~69, no.~3, pp.~1568--1583, 2024.

\bibitem{touraj2024tit}
T.~Soleymani, J.~S. Baras, and D.~G\"{u}nd\"{u}z, ``Transmit or retransmit: A tradeoff in networked control of dynamical processes over lossy channels with ideal feedback,'' {\em IEEE Trans. on Information Theory}, 2024.

\bibitem{uysal2022semantic}
E.~Uysal, O.~Kaya, A.~Ephremides, J.~Gross, M.~Codreanu, P.~Popovski, M.~Assaad, G.~Liva, A.~Munari, B.~Soret, T.~Soleymani, and K.~H. Johansson, ``Semantic communications in networked systems: A data significance perspective,'' {\em IEEE Network}, vol.~36, no.~4, pp.~233--240, 2022.

\end{thebibliography}
\bibliographystyle{ieeetr}

\end{document}